\documentclass{l4dc2024}

\usepackage{bbm}
\usepackage{dsfont}
\usepackage{algorithm}
\usepackage{algpseudocode}
\usepackage{xcolor}
\usepackage{graphicx}
\usepackage{float}

\usepackage[utf8]{inputenc}
\usepackage[english]{babel}
\usepackage{appendix}
\usepackage{wrapfig}
\usepackage{lettrine}
\usepackage{geometry}
\usepackage[compact]{titlesec}
\usepackage[font=small,skip=2pt]{caption}
\newtheorem{mylemma}{Lemma}
\usepackage{times}

\newcommand{\ie}{\textit{i}.\textit{e}.}

\title[Optimal control with probabilistic constraints]{Deterministic policy gradient based optimal control with probabilistic constraints}

\coltauthor{\Name{Arunava Naha} \Email{arunava.naha@angstrom.uu.se}\and \\ 
 \Name{Subhrakanti Dey} \Email{subhrakanti.dey@angstrom.uu.se}\\
 \addr Department of Electrical Engineering, Uppsala University, 75103 Uppsala, Sweden.}

\begin{document}
\setlength{\abovedisplayskip}{2pt}
\setlength{\belowdisplayskip}{2pt}

\maketitle
\vspace{-1.0cm}
\begin{abstract} 
	This paper studies a deep deterministic policy gradient (DDPG) based actor critic (AC) reinforcement learning (RL) technique to control a linear discrete-time system with a quadratic control cost while ensuring a constraint on the probability of potentially risky or undesirable events. The proposed methodology can be applied to both known and unknown system models with minor adjustments to the reward structure (negative cost). The problem is formulated by considering the average expected quadratic cost of the states and inputs over an infinite time horizon. Risky or undesirable events are represented as functions of the states at the next time step exceeding a user-defined limit. Two strategies are employed to manage the probabilistic constraint in scenarios of known and unknown system models. In the case of a known system model, the probabilistic constraint is replaced with an upper bound, such as the Chernoff bound. For unknown system models, the expected value of the indicator function of the occurrence of the risky or undesirable event is used. We have adopted a deterministic policy gradient (DPG) based AC method to derive a parameterised optimal policy. Extensive numerical simulations are performed using a second- and a fourth-order system, and the proposed method is compared with the standard risk-neutral linear quadratic regulator (LQR) and a chance-constrained model predictive control (MPC) method. The results demonstrate the effectiveness of the proposed approach in both known and unknown system model scenarios.
\end{abstract}

\begin{keywords}%
  optimal control, probabilistic constraints, chance constraints, deterministic policy based actor-critic method, reinforcement learning.
\end{keywords}

\section{Introduction} \label{sec:intro}
The problem of finding an optimal controller that minimizes the expected quadratic cost of states and inputs has been well-studied in the literature for linear time-invariant systems (LTI). The optimal control input for the control problem becomes a linear function of the states provided the noise has zero mean and a bounded second moment (\cite{bertsekasDynamicProgrammingOptimal}). Such a formulation of the cost is risk neutral since it only minimises the average value and does not consider the less frequent but risky or undesirable events. Such events may arise due to the presence of a long tail or a skewed distribution of the noise or uncertainty in the system. Therefore, for many practical systems, it is important to minimise the quadratic cost of states and inputs along with some additional constraints. For example, minimising the control cost of an unmanned aerial vehicle (UAV) is essential, and it is equally important for some applications to restrict the UAV from entering the adversarial range of vision (\cite{tsiamisRiskConstrainedLinearQuadraticRegulators2020}). Similarly, in the case of a climate-controlled building, the goal is to minimize the overall energy consumption while maintaining the temperature within a certain limit. Another practical example is the wind turbine, where the main objective is to maximize the power efficiency and, at the same time, to maintain the stress level on the blades, which varies due to the uncertain wind conditions, under a specified limit (\cite{schildbachScenarioApproachStochastic2014}). As studied by \cite{flemingStochasticMPCAdditive2019a}, the controller designed with a hard constraint is pessimistic compared to the one designed with a softer probabilistic constraint. In other words, the designed controller will lower the control cost if we constrain the probability of constraint violation to a small value, \i.e., allow the constraint to be violated only for very few occasions. For instance, temporarily exceeding specified temperature limits in a climate-controlled building under extreme conditions, like opening a window, may not significantly inconvenience occupants but could notably improve control costs.

\subsection{Related work}
A popular approach followed in the model predictive control (MPC) literature to design an optimal controller with probabilistic constraints, also known as chance constraints, involves scenario-based sampling. In the scenario-based sampling technique, samples are drawn from a known disturbance distribution and transform the probabilistic constraint into a set of finite algebraic conditions, as detailed in references like \cite{flemingStochasticMPCAdditive2019a,schildbachScenarioApproachStochastic2014}. Alternatively, \cite{hendriksDataControllerNonlinear2022} explored a different MPC approach, where the probabilistic chance constraint is approximated by an estimated expectation using the Hamiltonian Monte Carlo (HMC) method, specifically in the context of nonlinear systems. Other researchers have substituted the probabilistic chance or risk constraint with an algebraic one employing Chebyshev's inequality, as seen in works like \cite{schildbachLinearControllerDesign2015,yanStochasticModelPredictive2018a}. Furthermore, in (\cite{arcariStochasticMPCRobustness2023}), the chance constraints are formulated as the probability that the state and input values remain within certain sets, exceeding predefined thresholds. 

In addition to applying constraints on the probability of risky or undesirable events, a few other literature, such as \cite{zhaoGlobalConvergencePolicy2022,zhaoPrimaldualLearningModelfree2021}, have focused on constraining these events by limiting the conditional variance of a quadratic function of the states. These studies have demonstrated that when incorporating such a risk constraint along with the linear quadratic regulator (LQR) cost, the resulting optimal controller emerges as an affine function of the states (\cite{zhaoInfinitehorizonRiskconstrainedLinear2021,tsiamisRiskConstrainedLinearQuadraticRegulators2020}). To derive the optimal controller, these works employed the policy gradient based primal-dual optimization method, applicable to both model-based and model-free scenarios.

In scenarios where the underlying control system is either unknown or too complex to model accurately, reinforcement learning (RL) based algorithms have demonstrated remarkable usefulness. These algorithms are adept at finding optimal policies that maximise expected return over finite and infinite time horizons, as explored in various studies (\cite{bertsekasReinforcementLearningOptimal2019, busoniuReinforcementLearningControl2018,lopezEfficientOffPolicyQLearning2023}). Essentially, within the RL framework, an agent learns to make optimal decisions by using its own experiences or those of others without comprehensive knowledge about the environment or system (\cite{suttonReinforcementLearningSecond2018}). RL-based algorithms have also been studied to design the optimal controller that maximises the expected return along with additional constraints on risky or undesirable events (\cite{hanReinforcementLearningControl2021}). In (\cite{hanReinforcementLearningControl2021}), the constraint, which is called safety constraint, is modelled as the expected value of a continuous non-negative function of the states being within a specified threshold, and the optimal controller is derived using an actor-critic (AC) algorithm. 



\subsection{Our approach and contributions}
Our research addresses an optimal control problem focused on minimizing the average expected quadratic cost of states and inputs over an infinite time horizon subject to a constraint on the probability of risky or undesirable events. Risky or undesirable events are defined as occurrences in which a function of states exceeds a user-specified limit. Our approach differs from other studies in that we do not employ Chebyshev's inequality to replace probabilistic constraints (\cite{schildbachLinearControllerDesign2015, yanStochasticModelPredictive2018a}) or approximate them by sampling from known disturbance distributions (\cite{flemingStochasticMPCAdditive2019a, schildbachScenarioApproachStochastic2014}).

In addressing the probabilistic constraint, we adopt two distinct methodologies based on whether the system model is known or unknown. For cases where the system model is known, we substitute the probabilistic constraint with the Chernoff bound, which offers a tighter upper bound compared to the Chebyshev bound. Conversely, in scenarios with unknown system models, the probabilistic constraint is replaced by the expected value of an indicator function, indicating the occurrence of risky events. Here, we employ a model-free deterministic policy gradient (DPG) based AC method (\cite{lillicrapContinuousControlDeep2016}), utilizing two separate neural networks for the actor and value function approximation.

We compare our proposed policy with a risk-neutral LQR and a scenario-based MPC method by simulations. As anticipated, our policy effectively reduced the risky or undesirable events, albeit with a slight increase in quadratic cost compared to the standard LQR. While the MPC method showed performance comparable to our approach, its effectiveness is highly dependent on the length of the time horizon. Longer time horizons enhance MPC performance but also increase computational complexity. Moreover, MPC requires solving an optimization problem in real-time at each time step, unlike our proposed method. Once trained, our method relies solely on the actor-network as a feedforward network, significantly reducing runtime computational complexities. Additionally, our method offers a parameterized policy in the form of an actor-network, an advantage not available for the MPC method.


\subsection{Organization}
The rest of the paper is organized as follows. Section~\ref{sec:problem} formulates the optimisation problem mathematically. Section~\ref{sec:AC_method} discusses a solution approach to the optimisation problem using the AC method in the scenarios of known and unknown models. The reward structure used with the AC method is discussed in Section~\ref{sec:reward}. The analytical expressions of the probability bounds or the evaluation of the probability values for two special cases are derived in Section~\ref{sec:case_studies}. The results of the numerical simulations are provided and discussed in Section~\ref{sec:results}. Finally, Section~\ref{sec:conclusion} concludes the paper.  

\section{Problem Formulation} \label{sec:problem}
We consider the following linear time-invariant (LTI) system for the study. 
\begin{equation}
	{\bf{x}}_{k+1}={\bf A}{\bf{x}}_{k}+{\bf B}{\bf{u}}_{k}+{\bf{w}}_{k}.
	\label{eqn:state_eqn}
\end{equation}
Here ${\bf{x}}_{k}\in {\rm I\!R}^{n}$ and ${\bf{u}}_{k}\in {\rm I\!R}^{p}$ are the state and input vectors at the $k$-th time instant respectively, whereas ${\bf{w}}_{k} \in {\rm I\!R}^{n}$ is an independent and identically distributed  (iid) process noise with distribution $f_w({\bf w})$. ${\bf{A}}\in {\rm I\!R}^{n\times n}$, ${\bf{B}}\in {\rm I\!R}^{n\times p}$. 

We assume that all the states are measured and the system $({\bf A},{\bf B})$ is stabilizable. In the standard LQR problem, the following cost function is minimised. 
\begin{align}
J_c=\lim_{T\to\infty} \frac{1}{T} \sum_{k=1}^T E\left[ {\bf x}^T_k{\bf V}{\bf x}_k+{\bf u}^T_k{\bf U}{\bf u}_k\right],
	\label{eqn:cost_fun_J}
\end{align}
where ${\bf{V}}\in {\rm I\!R}^{n\times n}$ and ${\bf{U}}\in {\rm I\!R}^{p\times p}$ are positive definite weight matrices. We also assume that $({\bf A},{{\bf V}^{1/2}})$ is detectable. If we assume that the noise is zero mean and the second-order moment of the noise is bounded, then the optimum input appears as a fixed gain linear control signal (\cite{bertsekasDynamicProgrammingOptimal}), see (\ref{eqn:opt_u}). If the noise has a non-zero but known mean, we can always subtract the mean and get a zero mean noise.
\begin{align}
	{\bf u}^*_k&={\bf K}{ {\bf x}}_{k},  \label{eqn:opt_u} \ \\
	{\bf K}&=-\left( {\bf B}^T{\bf S}{\bf B}+{\bf U}\right)^{-1}{\bf B}^T{\bf S}{\bf A} \text{,} \label{eqn:L}
\end{align}
where $\bf S$ is the solution to the following algebraic Riccati equation,
\begin{align}
	{\bf S}={\bf A}^T{\bf S}{\bf A}+{\bf V}-{\bf A}^T{\bf S}{\bf B}\left({\bf B}^T{\bf S}{\bf B}+{\bf U}\right)^{-1}{\bf B}^T{\bf S}{\bf A}.
	\label{eqn:S}
\end{align} 
However, as discussed before, the cost function $J_c$ does not take into account the less frequent but risky events. Therefore, in our study, we use an additional constraint on the probability of risky or undesirable events, and the optimization problem takes the following form.  
\begin{equation}
\begin{aligned}
	& \min_u  J_c=\lim_{T\to\infty} \frac{1}{T} \sum_{k=1}^T E\left[ {\bf x}^T_k{\bf W}{\bf x}_k+{\bf u}^T_k{\bf U}{\bf u}_k\right] \label{eqn:opt_prob}\\
	& s.t. \lim_{T\to\infty} \frac{1}{T} \sum_{k=1}^T E\left[P\left\{f_c\left({\bf x}_{k+1} \right)\ge \epsilon \mid \Psi_{k} \right\}  \right] \le \delta. 
\end{aligned}
\end{equation}
Here, $\epsilon >0$ and $\delta >0$ are user selected parameters. $\Psi_{k} \triangleq \left \{{\bf x}_{k}, {\bf u}_{k} \mid k \ge 1 \right \}$ denotes the set of all information up to the instant $k$. 
\begin{remark}
	The risky or undesirable event is modelled as the function $f_c(\cdot)$ of the state crossing a threshold $\epsilon$, \ie, $f_c\left({\bf x}_{k+1} \right)\ge \epsilon$, and we are interested in limiting the probability of these events. Since such probability itself is a function of the random information set $\Psi_{k}$, we have taken the expectation with respect to this set in the above formulation. Furthermore, we are interested in keeping the long term average probability bounded over the infinite time horizon.
\end{remark}
The constrained optimization of (\ref{eqn:opt_prob}) is converted into the unconstrained stochastic control problem using the Lagrangian multiplier $C_l$ as follows, 
\begin{equation}
	\min_u J_L = \lim_{T\to\infty} \frac{1}{T} \sum_{k=1}^T E\left[g\left( {\bf x}_k,{\bf u}_k\right) \right],
	\label{eqn:JL}
\end{equation}
where the per stage cost $g(\cdot)$ takes the following form, 
\begin{align}
	&g\left( {\bf x}_k,{\bf u}_k\right) = f\left( {\bf x}_k,{\bf u}_k\right) +C_l  h_p\left( {\bf x}_k,{\bf u}_k\right)\text{, where} \label{eqn:gk} \\
	&f\left( {\bf x}_k,{\bf u}_k\right) = {\bf x}^T_k{\bf W}{\bf x}_k+{\bf u}^T_k{\bf U}{\bf u}_k,  \label{eqn:fk} \\
	& h_p\left( {\bf x}_k,{\bf u}_k\right) = P\left\{f_c\left({\bf x}_{k+1} \right)\ge \epsilon \mid \Psi_{k} \right\}.  \label{eqn:hpk} 
\end{align}

Note that the per-stage cost function may contain an intractable probability constraint, in general. The following section discusses how the probabilistic constraint can be converted into a more tractable constraint function.

\begin{remark}
A formal proof to show that the optimal solution to the cost function with a Lagrangian multiplier, as given in (\ref{eqn:JL}), will also be an optimal solution of the original constrained optimisation problem as in (\ref{eqn:opt_prob}) is difficult as such results for unbounded cost functions with possibly non-convex constraints are few. Also, in general, the optimal solution obtained may only be a local optimum.  Further investigation into this is currently underway.
\end{remark}

%
%
%
%
%
%


\section{Reformulation of the Constraint Function}\label{sec:reward}
In the sections that follow, we address the intractable probability constraint differently, depending on whether the model is known or unknown. For the study presented in this article, we have explored both a quadratic and a linear form for the function $f_c(\cdot) $ as follows. 
\begin{align}
    f_c\left({\bf x}_{k} \right) = {\bf x}^T_{k}{\bf Q}{\bf x}_{k},  \label{eqn:quad_const} \\
    f_c\left({\bf x}_{k} \right) = {\bf q}^T{\bf x}_{k}, \label{eqn:lin_const} 
\end{align}
where ${\bf Q}\in {\rm I\!R}^{n\times n} > 0$, and  ${\bf q}\in {\rm I\!R}^{n}$. Equation (\ref{eqn:lin_const}) is common in the literature on chance-constrained MPC. On the other hand, (\ref{eqn:quad_const}) is associated with the probability of exceeding a specified threshold for the quadratic cost in the subsequent state, thereby indicating the probability of an energy or cost outage.

\subsection{Known System Model} \label{subsec:known_AB}
For the discussion in this section, we assume that the system model, \ie, ${\bf A}$ and ${\bf B}$ matrices, is known. When $f_c\left({\bf x}_{k} \right)$ is a quadratic function of states (\ref{eqn:quad_const}), we replace the probability value in $g(\cdot)$, \ie, $h_p(\cdot)$ (\ref{eqn:hpk}), by the Chernoff bound $h_c(\cdot)$ as provided in the following Lemma~\ref{lemma:cher_bound}.
\begin{mylemma} \label{lemma:cher_bound}
For the LTI system given in (\ref{eqn:state_eqn}) and $f_c(\cdot)$ given in (\ref{eqn:quad_const}), the probability value $h_p(\cdot)$ (\ref{eqn:hpk}) will be upper bounded by $h_c(\cdot)$ as given in (\ref{eqn:hck}).
\begin{equation}
	h_c\left({\bf x}_k, {\bf u}_k\right)=\inf_{s\ge0}\left[e^{-\left(\epsilon-d_k \right)s}M\left(y_k,s \right) \right],
	\label{eqn:hck}
\end{equation}
where $M(y_k,s) \triangleq \text{E}\left[e^{sy_k}\right]$ denotes the moment generating function of $y_k$. We assume that the noise process ${\bf w}_k$ is such that $M(\cdot)$ exits. $y_k$ and $d_k$ are given as follows:
\begin{align}
	&y_k= {\bf w}^T_k{\bf Q}{\bf w}_k+{\bf a}^T_k{\bf w}_k \text{, and} \label{eqn:yk} \\
	& d_k = {\bf \hat x}^T_k{\bf Q}{\bf \hat x}_k, \label{eqn:dk}
\end{align}
where ${\bf a}_k = 2{\bf Q}{\bf \hat x}_k$ and ${\bf \hat x}_k$ is an estimate of ${\bf x}_{k+1}$ given $\Psi_k$ as defined below:
\begin{equation}
{\bf \hat x}_k = \text{E}\left[{\bf x}_{k+1}\mid\Psi_k \right]={\bf A}{\bf{x}}_{k}+{\bf B}{\bf{u}}_{k}.
\label{eqn:xhk}
\end{equation} 
\end{mylemma}
\begin{proof}
	Using (\ref{eqn:state_eqn}) and (\ref{eqn:xhk}), we can write, 
	\begin{equation}
		{\bf x}_{k+1} = {\bf \hat x}_k + {\bf w}_k.
		\label{eqn:xk_temp}
	\end{equation}
Then using (\ref{eqn:xk_temp}), the inequality in (\ref{eqn:hpk}) can be written in the following form, 
\begin{equation}
	\begin{aligned}
		&{\bf x}^T_{k+1}{\bf Q}{\bf x}_{k+1} \ge \epsilon 
		\Rightarrow &{\bf w}^T_k{\bf Q}{\bf w}_k+{\bf a}^T_k{\bf w}_k \ge \epsilon - d_k.
		\label{eqn:ineq_temp}
	\end{aligned}
\end{equation}
Finally, applying the Chernoff bound (\cite{pishro-nikIntroductionProbabilityStatistics2014}) on the conditional probability of the inequality (\ref{eqn:ineq_temp}), we get the bound $h_c(\cdot)$ in (\ref{eqn:hck}).
\end{proof}

\begin{remark}
	The Chernoff bound provides an upper limit of the probability value in the constraint. Therefore, an optimal solution obtained for the optimization problem using the per-stage cost with the Chernoff bound as given in (\ref{eqn:h_AB}) will also be a feasible solution to the original constrained problem with per-stage cost as given in (\ref{eqn:gk}).
\end{remark}

\begin{remark}
	The reason for using the Chernoff bound is that the Chernoff bound is tighter than the Chebyshev bound (\cite{pishro-nikIntroductionProbabilityStatistics2014}). 
	
\end{remark}
For the LTI system given in (\ref{eqn:state_eqn}) and $f_c(\cdot)$ given in (\ref{eqn:lin_const}), the probability value $h_p(\cdot)$ (\ref{eqn:hpk}) can be evaluated as
\begin{equation}
	h_c\left({\bf x}_k, {\bf u}_k\right)= \int_{{\mathcal{D}}_c} f_w({\bf w})dw.
	\label{eqn:hck_lin}
\end{equation}
Here, ${\mathcal{D}}_c = \left \{{\bf w} \mid {\bf q}^T{\bf w} \ge \epsilon - {\bf q}^T{\bf A}{\bf x}_k - {\bf q}^T{\bf B}{\bf u}_k \right \}$.
\begin{remark}
    In the case of Gaussian and Gaussian mixture models, represented by $f_w({\bf w})$, it is feasible to derive a closed-form expression for (\ref{eqn:hck_lin}). In cases where a closed-form solution is not achievable for a specific distribution, it is advisable to employ an appropriate bound, such as Markov's inequality, Chebyshev's inequality, Hoeffding's inequality and Chernoff's inequality, as a substitute for direct computation of the probability value.
\end{remark}

In conclusion, when system model parameters $\bf A$ and $\bf B$ are known, we have used the following per-stage cost function, see (\ref{eqn:h_AB}). Here $h_c\left( {\bf x}_k,{\bf u}_k\right)$ is given by (\ref{eqn:hck}) or (\ref{eqn:hck_lin}), where $f_c(\cdot) $ is a quadratic (\ref{eqn:quad_const}) or linear (\ref{eqn:lin_const}) function of the states, respectively.   
\begin{equation}
g\left( {\bf x}_k, {\bf u}_k \right) = f\left( {\bf x}_k,{\bf u}_k\right) +C_l  h_c\left( {\bf x}_k,{\bf u}_k\right).\label{eqn:h_AB}
\end{equation}
\subsection{Unknown System Model } \label{subsec:unknown_AB}
For the discussion in this section, we assume that the system model consisting of  ${\bf A}$ and ${\bf B}$ matrices, is unknown. Under that assumption, we can not evaluate $h_c(\cdot,\cdot)$  (\ref{eqn:hck}) or (\ref{eqn:hck_lin}). However, we can represent the probability value as an expectation as follows.
\begin{equation}
	P\left\{ f_c\left({\bf x}_{k+1}\right) \ge \epsilon \mid \Psi_{k}\right\} = \text{E} \left[\mathds{1}_{\left \{f_c\left({\bf x}_{k+1}\right) \ge \epsilon \right\}} \mid  \Psi_{k}  \right].
	\label{eqn:P2E}
\end{equation}
Here $\mathds{1}_{\left\{Z\right\}}$ denotes the indicator function where $Z$ denotes the condition. We can always evaluate the indicator function using the recorded data, even in cases where the system model is unknown. Consequently, for such cases, we adopt a slightly altered version of the per-stage cost function:
\begin{equation}
	g\left( {\bf x}_k, {\bf u}_k,{\bf x}_{k+1} \right) =  f\left( {\bf x}_k,{\bf u}_k\right) +C_l \mathds{1}_{\left\{f_c\left({\bf x}_{k+1}\right) \ge \epsilon\right\}}  
	\label{eqn:r_without_AB}
\end{equation}
Note that the expectation operator is already present for the cost function $J_L$ in (\ref{eqn:JL}). 

Up to this point, no specific assumptions have been made regarding the distribution of the noise ${\bf w}_k$. In other words, the noise may not be zero mean or Gaussian. However, the challenge lies in deducing the value of $h_c\left( {\bf x}_k,{\bf u}_k\right)$, as defined in (\ref{eqn:hck}) or (\ref{eqn:hck_lin}), for a scenario where the model is known. In the subsequent section, we have studied two specific cases: first, where the noise is assumed to be iid with a non-zero mean Gaussian distribution; and second, where the noise follows an iid Gaussian mixture model (GMM).
\section{Example cases: Evaluation of $h_c(\cdot)$ } \label{sec:case_studies}
In this section, we derive the analytical forms of $h_c\left( {\bf x}_k,{\bf u}_k\right)$ as given by (\ref{eqn:hck}) and (\ref{eqn:hck_lin}) for two special cases of $f_w({\bf w})$. 
\subsection{Case 1.1: $f_w({\bf w})$ is Gaussian and $h_c\left( {\bf x}_k,{\bf u}_k\right)$ given by (\ref{eqn:hck})   } \label{subsec:case1.1}
Here we assume that noise ${\bf w}_k$ is iid and ${\bf w}_k \sim\mathcal{N}\left(\mu_w,\Sigma_w \right)$, $\Sigma_w \in {\rm I\!R}^{n\times n} > 0$. Also, we assume that $\bf Q$ is symmetric. Under those conditions, we can derive an analytical expression of the moment generating function $M(\cdot)$ in (\ref{eqn:hck}) as stated in Theorem 3.2a.2 in (\cite{mathaiQuadraticFormsRandom1992}) and presented as the following Lemma. Replacing (\ref{eqn:M_gaussian}) in (\ref{eqn:hck}), we can evaluate $h_c\left( {\bf x}_k,{\bf u}_k\right)$.
\begin{mylemma} \label{lemma:gaussian_case}
	If ${\bf w}_k \in {\rm I\!R}^{n}$, iid, and ${\bf w}_k \sim\mathcal{N}\left(\mu_w,\Sigma_w \right)$, where $\Sigma_w \in {\rm I\!R}^{n\times n} > 0$, the moment generating function $M\left(y_k,s \right)$ of the random variable $y_k$ in (\ref{eqn:yk}) takes the following form, 
	\begin{equation}
		\begin{aligned}
			&M\left(y_k,s \right) = \exp\left(s\left(\mu_w^T{\bf Q}\mu_w + {\bf a}^T_k\mu_w\right)+ 0.5s\sum_{j=1}^nb_{k,j}^2\left(1-2s\lambda_j\right)^{-1}\right)\prod_{j=1}^n\left(1-2s\lambda_j\right)^{-1/2}. \label{eqn:M_gaussian}
		\end{aligned}
	\end{equation}
Here $\lambda_j$s are the eigenvalues of the matrix $\Sigma^{\frac{1}{2}}{\bf Q}\Sigma^{\frac{1}{2}}$, and $\bf P$ is the corresponding eigenvector matrix. ${\bf b}_k = \left[b_{k,1}, \cdots, b_{k,n} \right]^T = {\bf P}^T\left(\Sigma^{\frac{1}{2}}{\bf a}_k + 2\Sigma^{\frac{1}{2}}{\bf Q}\mu_w \right)$. ${\bf a}_k$ is the same as given in Lemma~\ref{lemma:cher_bound}. 
\end{mylemma}
\begin{proof}
	See the proof of Theorem 3.2a.2 from \cite{mathaiQuadraticFormsRandom1992}. \vspace{-0.5cm}
	\end{proof}
 \subsection{Case 1.2: $f_w({\bf w})$ is Gaussian and $h_c\left( {\bf x}_k,{\bf u}_k\right)$ given by (\ref{eqn:hck_lin})   } \label{subsec:case1.2}
 Here we assume that noise ${\bf w}_k$ is iid and ${\bf w}_k \sim\mathcal{N}\left(\mu_w,\Sigma_w \right)$, $\Sigma_w \in {\rm I\!R}^{n\times n} > 0$. Under these assumptions, ${\bf w}_q = {\bf q}^T{\bf w}$ becomes a Gaussian RV with mean ${\bf q}^T{\bf \mu}_w$ and variance ${\bf q}^T{ \Sigma}_w {\bf q}$. Finally, $h_c\left( {\bf x}_k,{\bf u}_k\right)$ (\ref{eqn:hck_lin}) can be evaluated as $\left( 1- CDF\left(\epsilon - {\bf q}^T{\bf A}{\bf x}_k - {\bf q}^T{\bf B}{\bf u}_k \right) \right)$. Here the function $CDF\left({\bf \tilde w}_{q} \right) = P\left\{{\bf w}_{q} \le {\bf \tilde w}_{q} \right\}$.

\subsection{Case 2.1: $f_w({\bf w})$ is GMM and $h_c\left( {\bf x}_k,{\bf u}_k\right)$ given by (\ref{eqn:hck})} \label{subsec:case2.1}
We assume the following LTI system model and the noise structure for this case. 
\begin{equation}
	{\bf{x}}_{k+1}={\bf A}{\bf{x}}_{k}+{\bf B}{\bf{u}}_{k}+{\bf B}{\bf{w}}_{k}.
	\label{eqn:state_eqn2}
\end{equation}
Here ${\bf w}_k \in {\rm I\!R}^{p}$, iid, and ${\bf w}_k \sim f_w({\bf w})$. Also, $f_w({\bf w})$ is a Gaussian mixture model (GMM) as follows,
\begin{equation}
	f_w({\bf w})=\sum_{j=1}^p \pi_j \mathcal{N}\left({\bf w};\mu_j,\Sigma_j \right).
	\label{eqn:f_gmm}
\end{equation}
Here $0< \pi_j < 1$ and $\sum_{j=1}^p \pi_j = 1$. Now the Chernoff bound is provided in the following lemma. 
\begin{mylemma} \label{lemma:cher_bound_g}
	For the LTI system given in (\ref{eqn:state_eqn2}), the probability value $h_p(\cdot)$ (\ref{eqn:hpk}) will be upper bounded by $h_g(\cdot)$ as given in (\ref{eqn:hgk}).
	\begin{equation}
		h_c\left({\bf x}_k, {\bf u}_k\right)=\inf_{s\ge0}\left[e^{-\left(\epsilon-d_k \right)s}M_g\left(y_{gk},s \right) \right],
		\label{eqn:hgk}
	\end{equation}
	where $M_g(\cdot)$ denotes the moment generating function of $y_{gk}$. $d_k$ is same as (\ref{eqn:dk}), and $y_{gk}$ is given as follows
	\begin{align}
		&y_{gk}= {\bf w}^T_k{\bf Q}_g{\bf w}_k+{\bf a}^T_{gk}{\bf w}_k \text{, and} \label{eqn:ygk} \\
		& {\bf Q}_g = {\bf B}^T{\bf Q}{\bf B}, \label{eqn:Bg}
	\end{align}
	where ${\bf a}_{gk} = 2{\bf B}^T{\bf Q}{\bf \hat x}_k$, and ${\bf \hat x}_k$ is an estimate of ${\bf x}_{k+1}$ given $\Psi_k$ as given in Lemma~\ref{lemma:cher_bound}. Finally, $M_g(\cdot)$  will take the following analytical form for the noise distribution given in (\ref{eqn:f_gmm}).
	\begin{equation}
		M_g\left(y_{gk},s \right)=\sum_{j=1}^p\pi_jM\left(y_{j,k},s \right).
		\label{eqn:Mg}
	\end{equation}
Here the function $M(\cdot)$ is same as given in Lemma~\ref{lemma:gaussian_case}. $y_{j,k}$ is given as follows,
\begin{equation}
	y_{j,k}= {\bf w}^T_{j,k}{\bf Q}_g{\bf w}_{j,k}+{\bf a}^T_{g,k}{\bf w}_{j,k} \text{, and } {\bf w}_{j,k} \sim \mathcal{N}\left( \mu_j, \Sigma_j \right) \label{eqn:yjk} 
\end{equation}
\end{mylemma}
\begin{proof}
	The proof of Lemma~\ref{lemma:cher_bound_g} is straightforward. (\ref{eqn:ygk}) and (\ref{eqn:Bg}) can be derived from (\ref{eqn:yk}) by replacing ${\bf w}_k$ by ${\bf B}{\bf w}_k$. (\ref{eqn:Mg}) is derived as follows (\cite{shahDistributionQuadraticForm2005}),
	\begin{equation}
	\begin{aligned}
	&M_g\left(y_{gk},s \right)=\int\exp\left(sy_{gk} \right)f_w\left( {\bf w}\right)ds =\sum_{j=1}^p \pi_j\int\exp\left(sy_{gk} \right)\mathcal{N}\left( {\bf w}; \mu_j,\Sigma_j\right)ds \text{ [using (\ref{eqn:f_gmm})]}\\
	&=\sum_{j=1}^p\pi_jM\left(y_{j,k},s \right) \nonumber
	\end{aligned}
	\end{equation} \vspace{-1.1cm}
	\end{proof}
 \subsection{Case 2.2: $f_w({\bf w})$ is GMM and $h_c\left( {\bf x}_k,{\bf u}_k\right)$ given by (\ref{eqn:hck_lin})} \label{subsec:case2.2} Here, we assume the system and noise models are the same as given in Case 2.1. From (\ref{eqn:state_eqn2}) and (\ref{eqn:f_gmm}), we can say ${\bf w}_q = {\bf q}^T{\bf B}{\bf w}$ will be a GMM with the following distribution, 
\begin{equation}
	f_{w_q}({\bf w}_q)=\sum_{j=1}^p \pi_j \mathcal{N}\left({\bf w}_q;{\bf q}^T{\bf B}\mu_j,{\bf q}^T{\bf B}{ \Sigma}_j{\bf B}^T {\bf q} \right).
	\label{eqn:fq_gmm}
\end{equation} 
Finally, $h_c\left( {\bf x}_k,{\bf u}_k\right)$ (\ref{eqn:hck_lin}) can be evaluated as 
\begin{equation}
	h_c\left( {\bf x}_k,{\bf u}_k\right)=\sum_{j=1}^p \pi_j \left( 1- CDF_j\left(\epsilon - {\bf q}^T{\bf A}{\bf x}_k - {\bf q}^T{\bf B}{\bf u}_k \right) \right).
	\label{eqn:fq_gmm}
\end{equation}  
Here, $CDF_j$ denotes the CDF for the normal distribution with mean ${\bf q}^T{\bf B}\mu_j$ and variance ${\bf q}^T{\bf B}{ \Sigma}_j{\bf B}^T {\bf q}$. 

In general, it is difficult to derive a closed-form solution for the stochastic sequential optimal control problem with infinite horizon, unless the value function has a differentiable closed-form expression. Therefore, in what follows, we have used a deep deterministic policy gradient (DDPG) based AC method \cite{lillicrapContinuousControlDeep2016}, since the state and action spaces are real-valued for our problem. Additionally, the AC method is a model-free approach, which is suitable for the unknown model scenario. Furthermore, AC methods provide a rich parameterized policy. Other than the DDPG based AC method, other policy gradient-based AC methods such as natural policy gradient, proximal policy gradient, etc. can also be used to solve the problem.
\section{DDPG based Actor-Critic (AC) Method} \label{sec:AC_method}
This section discusses the deep deterministic policy gradient (DDPG) based AC method (\cite{lillicrapContinuousControlDeep2016}) used in our study. DDPG algorithm incorporates separate target networks alongside a replay buffer, effectively diminishing the correlation between samples. Such an approach leads to a notably more stable training process. 

In our study, we want to derive a policy $\mu(\cdot)$ that maximizes the expected discounted return denoted by $Q^\mu$, see (\ref{eqn:Qfn}).
\begin{equation}
	Q^\mu \left({\bf x}_k, {\bf u}_k \right) = \text{E} \left[R_k \mid  {\bf x}_k, {\bf u}_k \right].
	\label{eqn:Qfn}
\end{equation}
Here $R_k$ denotes the discounted return starting from the $k$-th instant in time till the end of time as given in (\ref{eqn:return}).
\begin{equation}
	R_k = \lim_{T\to\infty} \sum_{i=k}^T \gamma^{i-k}r\left( {\bf x}_k, {\bf u}_k \right),
	\label{eqn:return}
\end{equation}
where $0 < \gamma <1$ is the discount factor and $r\left( {\bf x}_k, {\bf u}_k \right)$ is the reward if we use the control input ${\bf u}_k$ when the state of the system is ${\bf x}_k$ at $k$-th instant in time. We have formulated the reward as the negative per stage cost $g(\cdot)$ as in (\ref{eqn:h_AB}) and (\ref{eqn:r_without_AB}) for the known and unknown model scenarios, respectively.

The purpose of the discount factor is to balance the trade-off between immediate rewards and future rewards and to guide the agent towards actions that lead to long-term success. The discount factor also ensures the return $R_k$ to be finite when $T \rightarrow \infty$. In practice, we can select $\gamma$ to be close to unity, and comparing (\ref{eqn:return}) and (\ref{eqn:JL}), we can say $\text{E}[R_1] \rightarrow T*J_L$ as $\gamma \rightarrow 1$ and $r\left( {\bf x}_k, {\bf u}_k \right)  =-g\left( {\bf x}_k,{\bf u}_k\right) $, where $g\left( \cdot\right)$ is given by (\ref{eqn:h_AB}) or (\ref{eqn:r_without_AB}). 

For the AC method used in our study, the actor, which represents the policy ($\mu$) function, and the critic, which represents the quality ($Q(\cdot)$) function, are approximated by two neural networks parameterized by $\theta^\mu$ and $\theta^Q$, respectively. Furthermore, the actor network takes states as input and provides the control inputs. On the other hand, the critic network takes states and control inputs and provides the expected return from that state-action pair. The AC algorithm followed for our work is provided in Algorithm~\ref{algo:DDPT}, which is similar to the one studied in \cite{lillicrapContinuousControlDeep2016}. In DDPG, two separate target networks, $Q^t$ and $\mu^t$, for the critic and actor, respectively, are used to provide consistent targets to stabilize the learning process. The target networks are updated by slowly tracking the main networks as given in Algorithm~\ref{algo:DDPT}. 

\begin{algorithm}[h!]
\caption{DDPG Algorithm}
\label{algo:DDPT}
    Set $\tau$, learning rates, initial and final variances of zero mean Gaussian noise for exploration ($\mathcal N_t$).     

    
    Initialize critic network $Q({\bf x}, {\bf u} \mid \theta^Q)$ and actor network $\mu({\bf x} \mid \theta^\mu)$ with random weights $\theta^Q$ and $\theta^\mu$.
    
    Initialise the target network $Q^t$ and $\mu^t$ with weights $\theta^{Q^t} \leftarrow \theta^Q$ and $\theta^{\mu^t} \leftarrow \theta^\mu$. 
    
    Initialise the replay buffer $\mathcal D$.  
    
      \For{episode = 1, M}{
        Receive initial observation state ${\bf x}_1$.
        
        \For{t = 1, T}{
            Select action ${\bf u}_t = \mu({\bf x}_t \mid \theta^\mu) + \mathcal N_t$ [$\mathcal N_t$ is zero mean Gaussian noise].
            
            Execute action ${\bf u}_t$ and observe reward $r_t$ and observe new state ${\bf x}_{t+1}$.
            
            Store transition $\left({\bf x}_t, {\bf u}_t, r_t, {\bf x}_{t+1} \right)$ in $\mathcal D$.
            
            Sample a random minibatch of $N$ transitions $\left({\bf x}_i, {\bf u}_i, r_i, {\bf x}_{i+1} \right)$ from $\mathcal D$.
            
            Set $y_i = r_i + \gamma Q^t \left( {\bf x}_{i+1}, \mu^t \left( {\bf x}_{i+1} \mid \theta^{\mu^t} \right) \mid \theta^{Q^t}   \right)$.
            
            Update critic by minimizing the loss: $L = \frac{1}{N} \sum_i \left( y_i - Q \left( {\bf x}_i, {\bf u}_i \mid \theta^Q \right) \right)^2.$
            
            Update the actor policy using the sampled policy gradient:
            \begin{equation*}
                \triangledown_{\theta^\mu} J \approx \frac{1}{N} \sum_i \triangledown_u Q \left( {\bf x}, {\bf u} \mid \theta^Q \right)\mid_{{\bf x} = {\bf x}_i, {\bf u} = \mu({\bf x}_i)} \triangledown_{\theta^\mu} \mu \left({\bf x} \mid \theta^\mu \right)\mid_{{\bf x} = {\bf x}_i}.
            \end{equation*}
            Update the target networks:
            \begin{equation*}
                \theta^{Q^t} \leftarrow \tau \theta^Q + (1-\tau) \theta^{Q^t}\ \text{and}\  \theta^{\mu^t} \leftarrow \tau \theta^\mu + (1-\tau) \theta^{\mu^t}
            \end{equation*}
        }   
        Reduce the variance of Gaussian noise for exploration until it reaches its final value } \vspace{-0.1cm}
\end{algorithm} 
\begin{remark}
	The past experiences of the AC agent, stored as quadruples $\left({\bf x}_k, {\bf u}_k, r({\bf x}_k, {\bf u}_k), {\bf x}_{k+1} \right)$ in the replay memory, are used to train the networks. Furthermore, the reward function is only used for training. Once the actor network is trained, it only takes ${\bf x}_k$ as input and provides ${\bf u}_k$ as output. In other words, we do not need ${\bf x}_{k+1}$ to evaluate the control input value at the $k$-th instant, \ie, ${\bf u}_k$, for the unknown model case, which uses ${\bf x}_{k+1}$ to evaluate the reward value at $k$-th instant in time. Therefore, the control policy function or the actor remains causal for the unknown model case.  
\end{remark}

\begin{remark}
    The closed-loop system, while learning the optimal policy, may become unstable. In the literature, it is sometimes assumed that the system model is known, but the knowledge may not be perfect. This imperfect understanding is then leveraged to develop a stable controller, albeit not the most efficient one (\cite{brunke2022safe}).  In our research, we adopt a similar approach. Specifically, when the state of the system exceeds a predefined threshold, we temporarily transition to a previously established stable controller. The challenge of learning an optimal policy that concurrently maintains system stability during exploration phase is still an open problem and we leave it for our future work.
\end{remark}

\begin{remark}
The theoretical foundations of the DDPG algorithm find their roots in the deterministic policy gradient (DPG) algorithm, as introduced by \cite{silver2014deterministic}. \cite{silver2014deterministic} highlighted the critical importance of the compatibility between the policy and the critic function approximator for the DPG algorithm's convergence. In essence, the critic network must rapidly adapt to changes in the policy. This compatibility is crucial to obtain a true policy gradient in theory. Consequently, the DPG algorithm employs a linear function approximator for the critic network to facilitate theoretical convergence.

However, the transition to DDPG, which utilises deep neural networks to approximate both actor and critic, introduces complexities. This shift means that the aforementioned compatibility cannot be guaranteed. Despite this, empirical evidence from simulation studies like those of \cite{lillicrapContinuousControlDeep2016} shows stable learning in practice. More recent research, such as \cite{wang2021deep}, delves into the challenge of designing a compatible critic network specifically for the DDPG algorithm. However, the problem of theoretical convergence of DDPG algorithm is still an open problem and we leave it for our future studies.
\end{remark}
The following subsection discusses how to find an optimal value of the Lagrange multiplier $C_l$ (\ref{eqn:h_AB}) or (\ref{eqn:r_without_AB}) for a given probabilistic constraint threshold $\delta$ (\ref{eqn:opt_prob}). 
\subsection{Finding an optimal value of Lagrange multiplier $C_l$} 
Using a suitable search method, one can always find a finite positive $C_l$ (\ref{eqn:h_AB}) or (\ref{eqn:r_without_AB}) such that the given probabilistic constraint threshold $\delta$ (\ref{eqn:opt_prob}) is satisfied, unless the problem is infeasible, where $C_l$ becomes infinity. In this work, we have performed a grid search over a feasible range of values for $C_l$ and obtained a table of $C_l$ and corresponding estimated $\hat \delta$ values. One can always use a more precise algorithm for finding the optimal value of $C_l$ for a given $\delta$, such as a dual update algorithm based on a subgradient method with an appropriate step size as studied by \cite{yu2006dual}. The algorithm will converge as long as the step size is sufficiently small.

We have estimated $\hat \delta$ by Monte Carlo (MC) simulations for the known and unknown model scenarios. Say, we simulated $M$ number of trial runs, each for $N$ time steps, and we observed that the constraint $f_c(\cdot) \ge \epsilon$ got satisfied $n$ times in total, then $\hat \delta$ is estimated as $= \frac{n}{MN}$. Note that the CV (\%) variable used in Table~\ref{tab:2nd} and Table~\ref{tab:UAV} is evaluated as $=\frac{n}{MN} \times 100\%$.

In the following section, we have performed numerical simulations using two systems to investigate the performance of the proposed methods.
\section{Numerical Results} \label{sec:results}
For the numerical study, we have used a 2nd order open loop unstable LTI system model, which is a special case of Case  1 and a UAV model from \cite{zhaoPrimaldualLearningModelfree2021}, which is a special case of Case  2. The model parameters of the 2nd order system and the UAV are provided in Appendix~\ref{appdx:model_params}. For both models, we assume that full state information is available. We have compared our proposed method with the risk neutral standard LQR and a chance-constraint MPC method using two metrics, the average expected quadratic cost (\ref{eqn:cost_fun_J}) and the number of times the constraint was violated (CV \%). The optimal controller for the standard LQR problem becomes an affine function of the states for the UAV system, \ie, ${\bf u}_k = {\bf K} {\bf x}_k + {\bf l}$. ${\bf K}$  is the same as given in (\ref{eqn:L}), and ${\bf l}$ is evaluated as ${\bf l} = -(\pi_1\mu_1 + \pi_2\mu_2)$ (\cite{tsiamisRiskConstrainedLinearQuadraticRegulators2020}). On the other hand, we have used a scenario-based chance constraint MPC similar to the one studied by \cite{schildbachScenarioApproachStochastic2014} for our problem, see Algorithm~\ref{alg:chance_constr_MPC}. For the second-order system, a comparison is made between the proposed method and standard LQR; see Table~\ref{tab:2nd}. In Table~\ref{tab:2nd}, our comparison is limited to the proposed method and the standard LQR under the quadratic constraint scenario (Case 1.2). Although simulation studies were conducted for the linear constraint case and compared with MPC based method, these results were found to be analogous to those observed in the UAV model, and, therefore, omitted from the paper to save space. Performances of the proposed method, standard LQR and MPC based methods are compared for the UAV model in Table~\ref{tab:UAV}.
\begin{algorithm}[h!]
\caption{Scenario-based chance constraint MPC}
\label{alg:chance_constr_MPC}
Perform the following steps at each time step $t$:

    Measure the current state ${\bf x}_t$.

    Generate $S$ noise samples ${\bf w}_t^{(1)}, \cdots, {\bf w}_t^{(S)} \sim f_w({\bf w})$.

    Solve the following optimisation problem:
    \begin{equation}
        \begin{aligned}
            &\min_{{\bf u}_{1|t},\cdots,{\bf u}_{T|t}} \sum_{s=1}^S \sum_{i=1}^T \left(f\left( {\bf x}_{i|t}^{(s)},{\bf u}_{i|t}\right) + C_s \mathds{1}_{\left\{f_c\left({\bf x}_{i+1|t}^{(s)}\right) \ge \epsilon\right\}}   \right)\\
            &\text{s.t. (\ref{eqn:state_eqn}) is satisfied. }
        \end{aligned}
    \end{equation}

    Apply the first control input ${\bf u}_{1|t}$ to the system. \vspace{-0.1cm}
\end{algorithm}
Here, $f(\cdot)$ is the same as (\ref{eqn:fk}), and $f_c(\cdot)$ is given by (\ref{eqn:quad_const}) or (\ref{eqn:lin_const}). $C_s$ is the Lagrangian multiplier. ${\bf x}_{i|t}^{(s)}$ denotes the predicted state at $(i+t)$-th time intant for the $s$-th scenatio when the current state is ${\bf x}_t$. 
\begin{table} [h!]
\caption{2nd Order system \vspace{-.50cm}}
\begin{center}
\scriptsize
\begin{tabular}{|c|c|c|c|c|c|c|}
\hline
 & \multicolumn{2}{c|}{Proposed (Known Model)} & \multicolumn{2}{c|}{Proposed (Unknown Model)} & \multicolumn{2}{c|}{LQR} \\
&$Jc$ & CV (\%) & $Jc$ & CV (\%) & $Jc$ & CV (\%) \\
\hline
Case 1.2 & 74.35 & 8.85 & 70.79 & 10.13 & 68.01 & 13.49 \\
\hline
\end{tabular}
\end{center}
\label{tab:2nd}
\end{table} \vspace{-.50cm}
\begin{table} [h!]
\caption{UAV model}
\centering
\scriptsize
\begin{tabular}{|c|c|c|c|c|c|c|c|c|}
\hline
 & \multicolumn{2}{c|}{Proposed} & \multicolumn{2}{c|}{Proposed } & \multicolumn{2}{c|}{LQR} & \multicolumn{2}{c|}{MPC} \\
  & \multicolumn{2}{c|}{(Known Model)} & \multicolumn{2}{c|}{(Unknown Model)} & \multicolumn{2}{c|}{} & \multicolumn{2}{c|}{} \\

&$Jc$ & CV (\%) & $Jc$ & CV (\%) & $Jc$ & CV (\%) & $Jc$ & CV (\%) \\
\hline
Case 2.1 & 166.53 & 6.21 & 139.63 & 4.54 & 110.07 & 17.17 & 140.43 & 1.34\\
Case 2.2 & 129.98 & 4.0 & 122.3 & 1.0 & 110.07 & 13.0 & 128.2 & 0.9\\
\hline
\end{tabular}
\label{tab:UAV}
\end{table}
  Furthermore, for the simulation study, we have used the same network structure of two hidden layers of size (10,100) for both, actor and critic networks. Learning rate = 0.001 and $C_l = 100$ and $10$ for the quadratic and linear constraint cases, respectively. Also, we have used relu activation function for all the layers for both networks. For exploration, we have added a zero-mean Gaussian RV to the control input during training. The variation of the Gaussian RV is reduced from $5{\bf I}$ to $0.01{\bf I}$ in steps. For the MPC based method, the value of $C_s$ is chosen to be five and four for Case 2.1 and Case 2.2, respectively, so that the $J_c$ values becomes similar to the proposed method and we compared the CV values. Note that, in practice, a suitable value of the Lagrangian multiplier can be selected by performing a grid search over a feasible range of values.    

 As expected, the constraint violation percentage gets reduced for the proposed method at the expense of small increase in the quadratic cost when compared with the standard LQG. For instance, the constraint violation percentage got lowered by 64\% whereas $J_c$ increased by 34\% with respect to LQR for the UAV system under known model scenarios for the parameters used in the simulation as given in Appendix~\ref{appdx:model_params}. On the other hand, for the unknown model scenario, the constraint violation percentage is lowered by 74\% whereas $J_c$ increased by 27\% with respect to LQR for the same UAV system. It is interesting to note that the proposed method's performance under the unknown system model assumption is similar to the known model scenario. However, knowing the system model can be useful in designing an initial stable controller that will keep the system stable in the learning phase. This can be useful in the case of safety critical systems. On the other hand, for the unknown system model scenario, designing such an initial stable controller may be challenging. 

  When comparing our DDPG based method with the MPC based approach, we observe slightly better CV values for MPC at a similar control cost. However, it is crucial to note that MPC's performance is heavily depended on the chosen parameters \( S = 20 \) and \( T = 5 \). While increasing these parameters can enhance MPC's performance, it comes with the trade-off of increased computational complexity, which is of the order of $ST^2$ (\cite{skaf2009nonlinear}). In addition, MPC necessitates solving an optimization problem at every time step, in contrast to our method, which only requires evaluating the feed-forward actor network. This distinction renders our approach significantly less computationally complex than MPC. It is also important to acknowledge that MPC is a model-based method, which further differentiates it from our proposed technique. 
  
  \begin{wrapfigure}{tr}{0.3\textwidth}
  \centering
      \includegraphics[width=55mm]{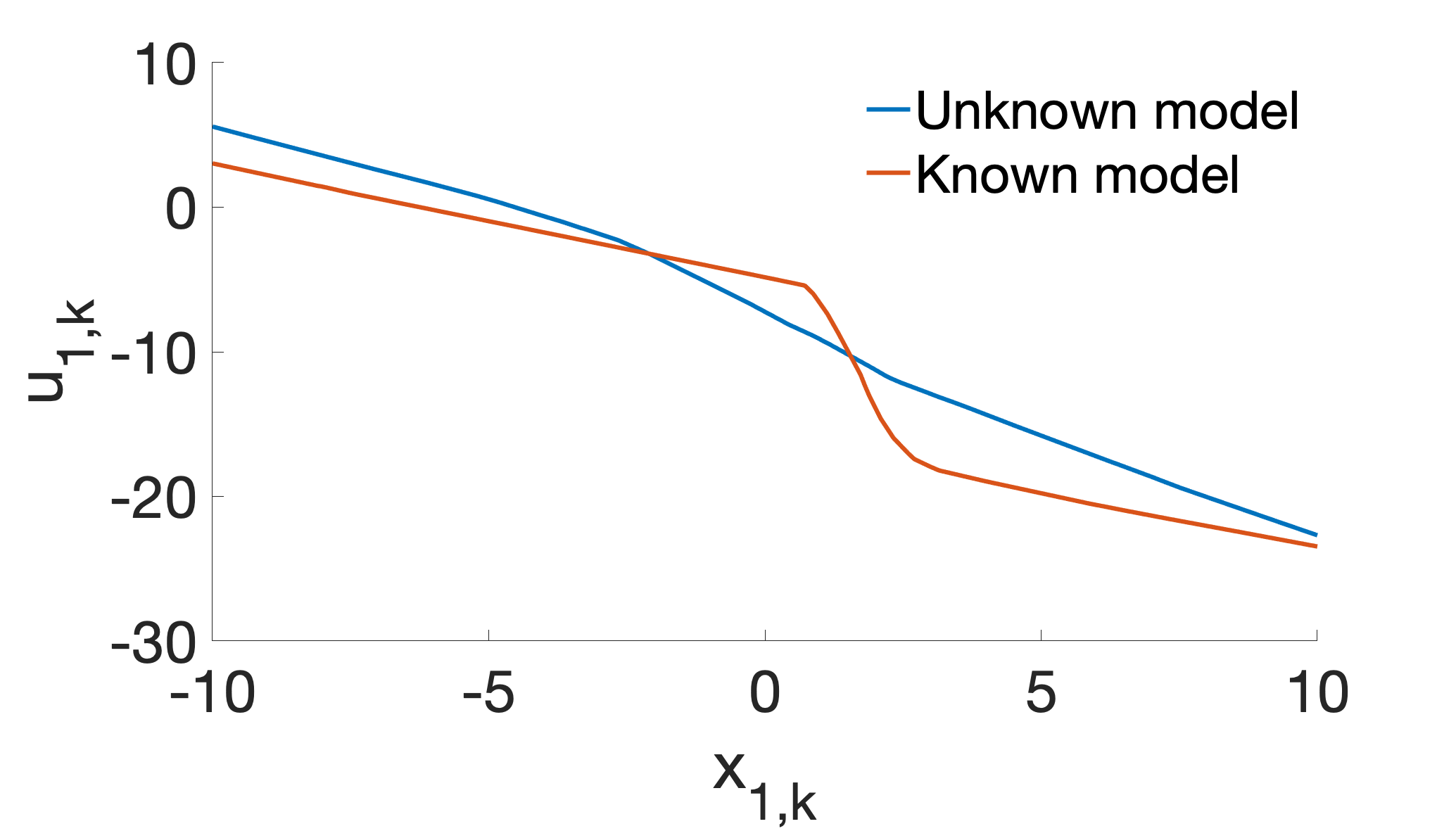}
     \caption{state vs. input}
     \label{fig:xvsu}
 \end{wrapfigure}
 In Fig.~\ref{fig:xvsu}, we have plotted a member of the state vector, \ie, ${\bf x}_{1,k}$, vs. a member of the control input vector, \ie, ${\bf u}_{1,k}$, for the known and unknown system model cases for the UAV system and quadratic constraint. To generate the plot, we have only varied  ${\bf x}_{1,k}$ and kept rest of the element of the state vector fixed at zero, and evaluated the control action from the actor network. While this particular plot only shows the behaviour of the control as a function of the state in one realisation, we have observed similar behaviour in many other instances. We believe this empirical study of the control law behaviour is of interest towards understanding the linearity or otherwise of an optimal control law.


\section{Conclusion} \label{sec:conclusion}
In conclusion, our study proposes a novel approach for handling probabilistic constraints in infinite-time horizon control problems for discrete-time linear Gauss-Markov systems, where risky or undesirable events are modelled as a function of the states crossing a user-defined limit. We have also studied a new reward structure for the case when the system model is unknown. Furthermore, we have applied a deterministic policy gradient-based actor-critic method (DDPG) to derive an optimal policy from observed data, and our extensive numerical simulations demonstrate the effectiveness of our approach in both known and unknown system model scenarios.  Our proposed approach has the potential to be applied in various real-world control problems where probabilistic constraints need to be handled effectively. In is interesting to note that the partial knowledge of the system model is not required to find a good policy but such knowledge can be useful to design an initial stable controller or a safe backup controller. Additionally, the MPC based approach performs better than the proposed method whereas the computational complexity is significantly higher. A formal proof of the convergence and stability properties of the proposed algorithm is under study and a topic for our future research.

\begin{appendices}
\section{Appendix} \label{appdx:model_params}
\textbf{Double inverted pendulum (2nd order system) model parameters:}
\begin{align*}
&{\bf A} =\begin{bmatrix}
	1.0  & 0.3  \\
	0.3 & 1.1
\end{bmatrix},   {\bf B} =\begin{bmatrix}
0.9 &, 0.5 \\
0.1 & 1.2
\end{bmatrix}, \Sigma_w = diag\left(2,2\right), {\bf W} = \begin{bmatrix}
	1.5 & 0.25  \\
	0.25 & 2.5
\end{bmatrix}, {\bf w} \sim \mathcal{N}\left({ 0}, \Sigma_w \right) \\
&{\bf U} = diag\left(40,70\right), {\bf Q} = 3{\bf W}, \epsilon = 95
\end{align*}
\textbf{UAV model parameters:}
\begin{align*}
	&{\bf A} =\begin{bmatrix}
		1 & 0.5 & 0 & 0  \\
		0 & 1 & 0& 0 \\
		0 & 0 & 1 & 0.5 \\
		0 &0 & 0 & 1
	\end{bmatrix}, {\bf B} =\begin{bmatrix}
	0.125 & 0  \\
	0.5 & 0 \\
	0 & 0.125 \\
	0 & 0.5
\end{bmatrix},  
	{\bf W} = diag\left(1, 0.1, 2, 0.2 \right), {\bf U} = {\bf I} , {\bf Q} = 2{\bf W}, \\
 & {\bf q} = \left[1, 0.1, 2, 0.2 \right]^T, \epsilon = 80 \text{ and } 5 \text{ for quadratic and linear constraints, respectively.}
\end{align*}
The noise vector ${\bf {\bar w}}_k$ consist of two elements, ${\bf {\bar w}}_k = [{\bf {\bar w}}_{1,k}, {\bf {\bar w}}_{2,k}]^T$, where ${\bf {\bar w}}_{1,k}$ has Gaussian mixture distribution, $\sim 0.2\mathcal{N}\left(3,30 \right)+0.8\mathcal{N}\left(8,60 \right)$, and ${\bf {\bar w}}_{2,k} \sim \mathcal{N}\left(0,0.01 \right)$. We can easily convert the noise model for this system into the one studied as Case Study 2 (\ref{eqn:f_gmm}) as follows.
\begin{align*}
 	\pi_1 = 0.2, \pi_2 = 0.8, \mu_1 = \begin{bmatrix} 3 & 0 \end{bmatrix}^T, \mu_2 = \begin{bmatrix} 8 & 0 \end{bmatrix}^T, 
 	\Sigma_1 = diag(30, 0.01), \Sigma_2 = diag(60, 0.01).
 \end{align*}


\end{appendices}
\acks{We thank Swedish Research Council for supporting this work. Grant no. 2017-04053.}

\bibliography{Const_Control}

\end{document}